\documentclass{article}
\usepackage{xcolor}
\usepackage{graphicx}
\usepackage{algorithm2e}
\usepackage{amssymb}
\usepackage{amsmath}
\usepackage{subfig}
\usepackage{authblk}
\usepackage{float}
\usepackage{placeins}
\usepackage[margin=2.9cm]{geometry}
\usepackage[T1]{fontenc}
\usepackage{multirow}
\usepackage{amsthm}
\newtheorem{theorem}{Theorem}[section]
\newtheorem{corollary}[theorem]{Corollary}
\newtheorem{lemma}[theorem]{Lemma}
\newtheorem{assumption}{Assumption}

\title{Diffusion Methods for Generating Transition Paths \thanks{This work is supported in part by National Science Foundation via award NSF DMS-2012286 and DMS-2309378.}} 
\author{Luke Triplett, Jianfeng Lu}
\affil{Duke University}
\date{September 18, 2023}

\begin{document}
\maketitle

\begin{abstract}
In this work, we seek to simulate rare transitions between metastable states using score-based generative models. An efficient method for generating high-quality transition paths is valuable for the study of molecular systems since data is often difficult to obtain. We develop two novel methods for path generation in this paper: a chain-based approach and a midpoint-based approach. The first biases the original dynamics to facilitate transitions, while the second mirrors splitting techniques and breaks down the original transition into smaller transitions. Numerical results of generated transition paths for the M\"uller potential and for Alanine dipeptide demonstrate the effectiveness of these approaches in both the data-rich and data-scarce regimes.
\end{abstract}

\section{Introduction}

A challenge arises in the study of molecular dynamics when the behavior of a system is characterized by rare transitions between metastable states. Practically, these rare transitions mean that Monte Carlo simulations take a prohibitively long time even if enhanced sampling techniques are used. 

A common way to understand the transition process is to sample transition paths from one metastable state to the other and use this data to estimate macroscopic properties of interest. Sampling through direct simulation is inefficient due to the high energy barrier, which has led to the exploration of alternative methods. Some notable methods include transition path sampling \cite{bolhuis2002transition}, biased sampling approaches \cite{pinski2010transition}, and milestoning \cite{faradjian2004computing}. Across different applications of rare event simulation, importance sampling and splitting methods are notable. The former biases the dynamics to reduce the variance of the sampling \cite{fleming1977exit}, while the latter splits rare transitions into a series of higher probability steps \cite{cerou2019adaptive}.

Broadly, generative models such as Generative Adversarial Models (GANs) and Variational Autoencoders (VAEs) have been developed to learn the underlying distribution of a dataset and can generate new samples that resemble the training data, with impressive performance. Among other tasks, generative models have been successfully used for image, text, and audio generation. Recently, VAEs have been used for transition path generation \cite{lelievre2023generative}. For this approach, the network learns a map from the transition path space to a smaller latent space (encoder) and the inverse map from the latent space back to the original space (decoder). The latent space is easier to sample from and can be used to generate low-cost samples by feeding latent space samples through the decoder.

Another promising method from machine learning frequently used in image-based applications is diffusion models or score-based generative models \cite{song2020score} \cite{croitoru2023diffusion}. This paper will provide methods to generate transition paths under overdamped Langevin dynamics using score-based generative modeling. As with VAEs, diffusion models rely on a pair of forward and backward processes. The forward process of a diffusion model maps the initial data points to an easy-to-sample distribution through a series of noise-adding steps. The challenge is to recover the reverse process from the noisy distribution to the original distribution, which can be achieved by using score matching \cite{hyvarinen2005estimation} \cite{vincent2011connection}. Having estimated the reverse process, we can generate samples from the noisy distribution at a low cost and transform them into samples from the target distribution.

The naive application of score-based generative modeling is not effective because of the high dimensionality of discretized paths. We introduce two methods to lower the dimensionality of the problem, motivated by techniques that have been used previously for simulating transition paths. The chain method, which is introduced in Section \ref{chain_gen_sec}, updates the entire path together and relies on a decomposition of the score in which each path point only depends on adjacent path points. The midpoint method, which we outline in Section \ref{midpoint_generation_sec}, generates path points separately across multiple iterations. It also uses a decomposition which gives the probability of the midpoint of the path conditioned on its two endpoints.

In Section \ref{reverse SDE section}, we give an overview of diffusion models and the reverse SDE sampling algorithm for general data distributions. In Section \ref{transition_paths_overview}, we give background on transition paths and transition path theory, before discussing the two methods for generating transition paths in Section \ref{generating transition paths}. Numerical results and a more detailed description of our algorithm are included in Section \ref{results}. The main contributions of this paper are to establish that diffusion-based methods are effective for generating transition paths and to propose a new construction for decomposing the probability of a transition path for dimension reduction.

\section{Reverse SDE Diffusion Model} \label{reverse SDE section}

There are three broad categories commonly used for diffusion models. Denoising Diffusion Probabilistic Models (DDPM) \cite{ho2020denoising} noise the data via a discrete-time Markov process with the transition kernel given by $P(x_t|x_{t-1}) = \mathcal{N}(x_t; \sqrt{1 - \beta_t} x_{t-1}, \beta_t I)$, where the hyperparameters $\{\beta_t\}$ determine the rate of noising. The reverse kernel of the process is then approximated with a neural network. The likelihood $P(x_t|x_{t-1})$ can't be calculated, so a variational lower bound for the negative log-likelihood is used. NCSM \cite{song2019generative} uses a different forward process, which adds mean-0 Gaussian noise. NCSM uses annealed Langevin sampling for sample generation, where the potential function is an approximation of the time-dependent score function $\nabla \log(p_t(x))$. The third approach models the forward and reverse processes as an SDE \cite{song2020score}. The model used in this paper can be described in either the DDPM or the reverse SDE framework, but we will describe it in the context of reverse SDEs.

Let us first review the definitions of the forward and backward processes, in addition to describing the algorithms for score-matching and sampling from the reverse SDE. Reverse SDE diffusion models are a generalization of NCSM and DDPM. As for all generative modeling, we seek to draw samples from the unknown distribution $p_{data}$, which we have samples from. In the case of transition paths, the data is in the form of a time series. Since diffusion models are commonly used in computer vision, this will often be image data. The forward process of a reverse SDE diffusion model maps samples from $p_{data}$ to a noisy distribution using a stochastic process $x_t$ such that
\begin{equation}
    dx_t = f(x_t, t) dt + g(t) dW_t, t \in [0, T],
    \label{stoc_process}
\end{equation}
where $W_t$ is standard Brownian motion. We denote the probability density of $x_t$ as $p_t(x)$. In this paper, we will use an Ornstein-Uhlenbeck forward process with $f(x_t, t) = -\beta x_t, g(t) = 1$. Then, $x_t \mid x_0 \sim N(x_0 e^{-\beta t}, \frac{1}{2 \beta} (1 - e^{-2 \beta t}) I)$ and $x_t \mid x_0 \overset{d}{\longrightarrow} N(0, \frac{1}{2 \beta}I)$ as $t \to \infty$. This means that as long as we choose a large enough $T$, we can start the reverse process at a standard normal distribution. The corresponding reverse process is described as follows:
\begin{equation}
\label{eqn:1.2}
    dx_t = (f(x_t, t) - g(t)^2 \nabla \log {p_t(x_t)}) d\tilde t + g(t) d\tilde W_t, t \in [0, T],
\end{equation}
where time starts at $T$ and flows backward to $0$, i.e., "$d\tilde t$" is a negative time differential. Reversing the direction of time we can get the equivalent expression
\begin{equation}
    dx_t = (-f(x_t, T-t) + \nabla \log {p_{T-t}(x_t)} g(x_t, T-t)^2) dt + g(x_t, T-t) dW_t.
\end{equation}
$\nabla \log {p_t(x)}$ is the score function of the noised distribution, which can't be retrieved analytically. So, the computational task shifts from approximating the posterior $p(x)$ directly, which is the target of energy-based generative models, to approximating the noised score $\nabla \log {p_t(x)}$. Posterior estimation of complex distributions is a well-studied and challenging problem in statistics. Modeling the score is often more tractable and does not require calculating the normalizing constant.

\subsection{Approximating the Score Function}
In reverse SDE denoising, we use a neural network to parameterize the time-dependent score function as $s_{\theta}(x, t)$. Since the reverse SDE involves a time-dependent score function, the loss function is obtained by taking a time average of the distance between the true score and $s_\theta$. Discretizing the SDE with time steps $0 = t_0 < t_1 < ... < t_N = T$, we get the following loss function:
\begin{equation}
L(\theta) = \frac{1}{N} \sum_{i = 1}^{N} h_{t_i} \mathbb{E}_{x_{t_i}}||\nabla \log p_{t_i} (x_{t_i}) - s_\theta(x_{t_i}, t_i)||_2^2.
\end{equation}

Using the step size $h_{t_i} = t_{i} - t_{i-1}$ is a natural choice for weighting, although it is possible to use different weights. The loss function can be expressed as an expectation over the joint probability of $x_{t_i}, x_0$ as \cite{vincent2011connection}
\begin{equation}
L(\theta) = \frac{1}{N} \sum_{i = 1}^{N} h_{t_i} \mathbb{E}_{x_0} \mathbb{E}_{x_{t_i} \mid x_0} ||\nabla \log p_{t_i} (x_{t_i} \mid x_0) - s_\theta(x_{t_i}, t_i)||_2^2.
\end{equation}
We can calculate $\nabla \log p_{t_i} (x_{t_i} \mid x_0)$ based on the forward process (in the case of Ornstein-Uhlenbeck, $p_{t_i} (\cdot \mid x_0)$ is a Gaussian centered at $x_0 e^{-\beta t}$). Remarkably, this depends only on $s_\theta$, $t$, and the choice of forward SDE. A similar loss function is used for NCSM with a different noise-adding procedure. The training procedure follows from the above expression for the loss.

\RestyleAlgo{ruled}
\SetKwComment{Comment}{/* }{ */}

\begin{algorithm}[H]
\caption{Learning the Score Function}\label{alg:two}
\KwData{$\{x_i\}_{i=1}^M, \{(t_i, h_i)\}_{i=1}^N, nb$ \Comment*[r]{nb = num batches, $h_i$ = step size}}
\KwResult{$s_\theta (x)$}
$i \gets 0$\;
$j \gets 0$\;
\For{$i < nb$}{
  $x_b \gets $ random.choice($x, M / nb$) \;
  \For{$j < N$} {
  $noise \gets randn\_like(x_b)$ \;
  $\sigma_j \gets \sqrt{\frac{1}{2 \beta} (1 - e^{-2 \beta})}$ \;
  $\tilde x_b \gets x_b e^{-\beta t} + \sigma_j \cdot noise$ \;
  $loss \gets loss + h_j \cdot ||s_\theta(\tilde x_b, \sigma_j) + \frac{x_b e^{-\beta t} - \tilde x_b}{\sigma_j^2}||_2^2$ \;
  $j \gets j + 1$\;
  }
  $loss$.backwards() \;
  $i \gets i + 1$\;
}
\end{algorithm}

\subsection{Sampling from Reverse SDE} \label{reverse_SDE_sampling}

Once we have learned the time-dependent score, we can use it to sample from the original distribution using the reverse SDE. Substituting our parameterized score function $s_\theta$ into \eqref{eqn:1.2}, we get an approximation of the reverse process,
\begin{equation}
d \tilde x_t = (-\beta x_t - s_\theta(x_{t}, t)) d\tilde t + d\tilde W_t, t \in [0, T].
\end{equation}

We want to evaluate $s_\theta$ at the times at which it was trained, so we discretize the SDE from equation \eqref{eqn:1.2} using the sequence of times $0 = \tilde t_0 \leq \tilde t_1 \leq ... \leq \tilde t_N = T$, where $\tilde t_k = T - t_{N-k}$. In forward time (flowing from 0 to T), this discretization gives
\begin{equation}
 d \tilde x_t = (\beta x_t + s_\theta(x_{\tilde t_k}, T - \tilde t_k)) dt + d W_t, t \in [\tilde t_k, \tilde t_{k+1}].
 \end{equation}
 We solve for $x_t$ using It\^{o}'s formula to retain the continuous dynamics of the first term. Let $\tilde z_t = e^{-\beta t} \tilde x_t$ and $\eta_k$ be a standard Gaussian random variable, then
\begin{equation}
d \tilde z_t = (-\beta e^{- \beta t} \tilde x_t + (\beta e^{-\beta t} \tilde x_t + e^{-\beta t} s_\theta(\tilde x_{\tilde t_k}, T - \tilde t_k)) dt + e^{-\beta t} dW_t.
\end{equation}
After solving for $\tilde z_{\tilde t_k}$, we can get an explicit expression for $\tilde x_{\tilde t_k}$, which is known as the exponential integrator scheme \cite{zhang2022fast}. 
\begin{equation}\label{eqn:2.3}
\tilde x_{\tilde t_{k+1}} = e^{\frac{1}{2} (\tilde t_{k+1} - \tilde t_k)} \tilde x_{\tilde t_k} + 2 (e^{\frac{1}{2} (\tilde t_{k+1} - \tilde t_k)} - 1) s_\theta(\tilde x_{\tilde t_k}, T - \tilde t_k) + \sqrt{e^{\tilde t_{k+1} - \tilde t_k} - 1} \cdot \eta_k,
\end{equation}
where $\eta_k \sim N(0, I_d)$. We denote the distribution of $x_{t_k}$ as $q_{t_k}$. According to the literature \cite{chen2023improved}, using an exponentially decaying step size for the discretization points provides strong theoretical guarantees for the distance between $q_{t_N}$ and the true data distribution, which we will explore further in Section \ref{conv-guarantees}. For the forward process, this means that $\frac{t_k - t_{k-1}}{t_{k+1} - t_k} = M$. It is important to choose $t_{min}, M$, and $N$ carefully, as they significantly affect the performance. More details about our implementation can be found in Section \ref{results}. Having estimated the score function, we can generate samples from a distribution that is close to $p_{data}$ by setting $\tilde x_0 \sim p_T \approx N(0, \frac{1}{2\beta} I)$, then repeatedly applying \eqref{eqn:2.3} at the discretization points. A visualization of the distribution of $\tilde x$ at different $\tilde t_k$ values under this procedure is shown in Figure \ref{time evolution}.

\begin{figure}[H]
  \centering
  \subfloat{\includegraphics[width=0.21\textwidth]{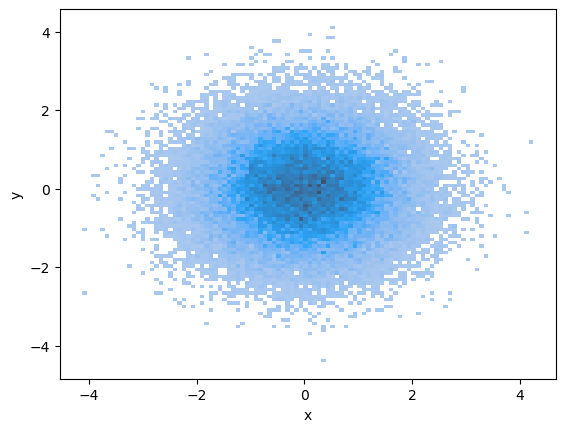}}%
  \qquad
  \subfloat{\includegraphics[width=0.21\textwidth]{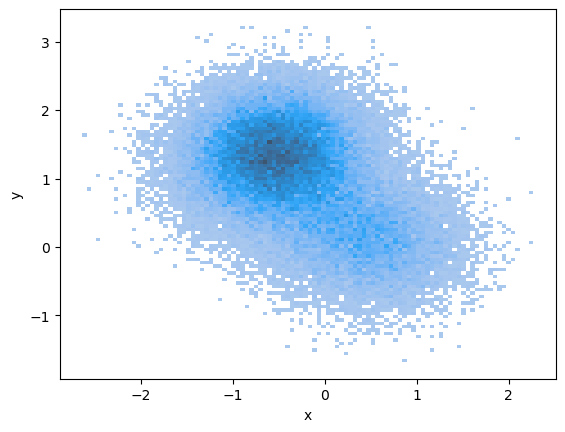}}%
  \qquad
  \subfloat{\includegraphics[width=0.21\textwidth]{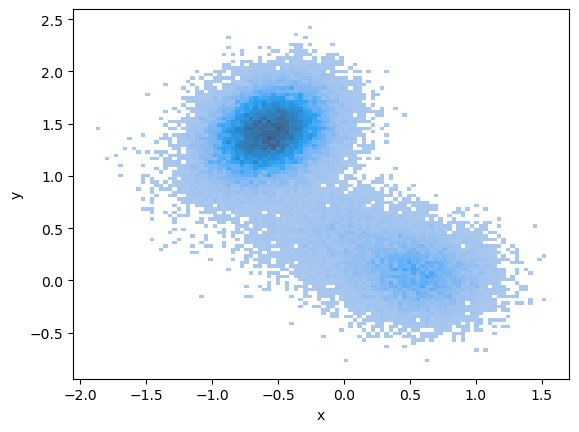}}%
  \qquad
  \subfloat{\includegraphics[width=0.21\textwidth]{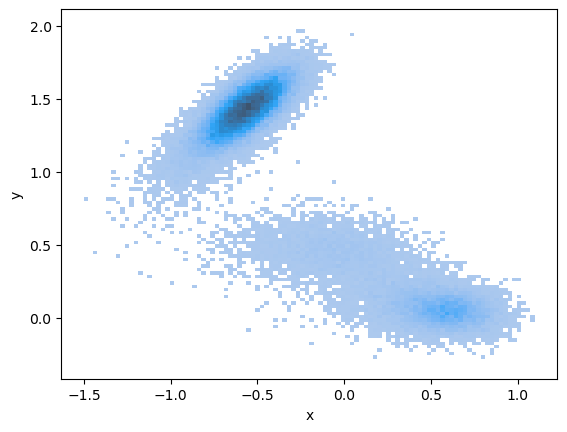}}%
  \caption{\textbf{Time evolution of the distribution of samples} using exponential integrator scheme from \eqref{eqn:2.3}. In the leftmost figure at $t = 0$, the samples are from a Gaussian distribution centered at (0,0). In the rightmost figure at $t = T$, the samples are approximately from the distribution $p(x) = e^{-V(x)}$, where $V$ is the M\"uller potential.} \label{time evolution}
\end{figure}

\section{Transition Paths} \label{transition_paths_overview}
\subsection{Overdamped Langevin Dynamics}
Transitions between metastable states are crucial to understanding behavior in many chemical systems. Metastable states exist in potential energy basins and so transitions will often occur on a longer time scale than the random fluctuations within the system. Due to their infrequency, it is challenging to effectively describe these transitions. In the study of transition paths, it is often useful to model the governing system as an SDE. In this section, we will look at the overdamped Langevin equation, given by
\begin{equation}
    dX_t = -\nabla V(X_t)dt + \sqrt{2B^{-1}}dW_t, \quad X_0 = x \in \mathbb{R}^d,
    \label{eq:1.1}
\end{equation}
where $V(x)$ is the potential of the system and $W_t$ is d-dimensional Brownian motion. Let $\mathcal{F}_t = \sigma({X_s, 0 \leq s \leq t})$ be the filtration generated by $X_t$. In chemical applications, $B^{-1}$ is the temperature times the Boltzmann constant. For reasonable $V$, the invariant probability distribution exists and is given by
\begin{equation}
    p(x) = e^{-\beta \nabla V(x)}/Z, \qquad  Z = \int_\Omega e^{-\beta \nabla V(x)} dx.
\end{equation}
We can see an example of samples generated from the invariant distribution when $V$ is the M\"uller potential in the above figure.

Consider two metastable states represented by closed regions $A, B \subset \mathbb{R}^n$. Their respective boundaries are $\partial A, \partial B$.
The paths that go from the boundary of $A$ to the boundary of $B$ without returning to $A$ are called transition paths \cite{vanden2010transition}. This means that a realization $\{X_s\}_{s=0}^T$ is a transition path if $X_0 \in \partial A, X_T \in \partial B$, and $X_s \notin A \cup B, \ \forall \ s \in (0, T)$. The distribution of transition paths is the target distribution of the generative procedure in our paper. A similar problem involves trajectories with a fixed time interval and is known as a bridge process. Much of the analysis in this paper can be extended to the fixed time case by removing the conditioning on path time during training and generation.

An important function for the study of transition paths is the committor function, which is defined by the following boundary value problem \cite{vanden2010transition}, \cite{lu2015reactive}
\begin{equation}
\begin{cases}
        L^P q(x) = 0 & \text{if } x \notin A \cup B\\
        q(x) = 0 & \text{if } x \in A\\
        q(x) = 1 & \text{if } x \in B
    \end{cases},
\label{committor}
\end{equation}
where $q(x)$ is the committor and $L^P$ is the generator for \eqref{eq:1.1}, defined as $L^P f = B^{-1} \Delta f - \nabla V \nabla f$.  

\subsection{Distribution of Transition Paths} \label{path_dist}
We will now examine the distribution of transition paths as in \cite{lu2015reactive}. Consider stopping times of the process $X_t$ with respect to $\mathcal{F}_t$:
\begin{align}
    \tau^{(X)}_A = \inf \{s \geq t : s \in A \}, \\
    \tau^{(X)}_B = \inf \{s \geq t : s \in B \}.
\end{align}
Let $E$ be the event that $\tau^{(X)}_B < \tau^{(X)}_A$. We are only looking at choices of $V(x)$ such that $P(\tau^{(X)}_A < \infty) = 1$ and $P(\tau^{(X)}_B < \infty) = 1$. It follows from It\^{o}'s formula that $q(x) = P(E \mid X_0 = x)$ is a solution to \eqref{committor}. Thus, the committor gives the probability that a path starting from a particular point reaches region $B$ before region $A$.

Consider the process $Z_t = X_{t \wedge \tau_A \wedge \tau_B}$, the corresponding measure $\mathbb{P}_x$, and the stopping times $\tau_A^{(Z)}, \tau_B^{(Z)}$. Since the paths we are generating terminate after reaching $B$, we will work with the $Z$ process, though it is possible to use the original $X$ process as well. We define $E^*$ as the event that $\tau_B^{(Z)} < \infty$. The function defined as $q(x) = P(E^* \mid Z_0 = x)$ is equivalent to the committor for the $X$ process. We are interested in paths drawn from the measure $\mathbb{Q}_x = \mathbb{P}_x(\cdot \mid E^*)$ on transition paths. The Radon-Nikodym derivative is given by
\begin{equation}
    \frac{d \mathbb{Q}_x}{d \mathbb{P}_x} = \frac{\mathbb{I}_E^*}{q(x)}.
\label{rn_derivative}
\end{equation}

Suppose that we have a path $\{ Z_s \}_{s=0}^\infty$ starting at $x$ and ending in $B$. Eq.~\eqref{rn_derivative} states that the relative likelihood of $\{ Z_s \}$ under $\mathbb{Q}_x$ compared to $\mathbb{P}_x$ increases as $q(x)$ decreases. This follows our intuition, since as $q(x)$ decreases, a higher proportion of paths starting from $x$ from the original measure will end in $A$ rather than $B$.

Using Doob's h-Transform \cite{day1992conditional}, we have that
\begin{align*}
    & P(Z_{t+s} = y \mid Z_t = x, E^*) = P(Z_{t+s} = y \mid Z_t = x)\frac{q(y)}{q(x)} \\
    \implies & L^Q f = \frac{1}{q} L^P (q f) \\
    \implies & L^Q f = L^P f + \frac{2 B^{-1} \nabla q}{q} \cdot \nabla f,
    \label{doob}
\end{align*}
and the law of $Y_t$ given by the SDE
\begin{equation}
    dY_t = (-\nabla V(Y_t) + 2B^{-1} \frac{\nabla q(Y_t)}{q(Y_t)})dt + \sqrt{2 B^{-1}} dw_t
    \label{transition_path_SDE}
\end{equation}
is equivalent to the law of transition paths. Thus, conditioning on transition paths is equivalent to adding a drift term to \eqref{eq:1.1}.

\section{Generating Transition Paths with Diffusion Models}\label{generating transition paths}

We can generate transition paths by taking sections of an Euler-Maruyama simulation such that the first point is in basin A, the last point is in basin B, and all other points are in neither basin. This discretizes the definition of transition paths discussed earlier. We will represent transition paths as $x = \{x_i\}_{i=1}^m, x_i \in \mathbb{R}^d$, where $m$ is the number of points in the path. We will denote the duration of a particular path as $T^*$. This should not be confused with $T$, which represents the duration of the diffusion process for sample generation. It is convenient to standardize the paths so that they contain the same number of points, and there is an equal time between all subsequent points in a single path.

\subsection{Chain Reverse SDE Denoising}
\label{chain_gen_sec}

In this section, we introduce chain reverse SDE denoising, which learns the gradient for each point in the path separately by conditioning on the previous point. Specifically, we will use $\frac{\partial}{\partial x_n} \log p_t(x(t))$ to represent the component of score function which corresponds to $x_n$ and $s^*(x(t), t)^{(n)}$ as the corresponding neural network approximation. Unlike in the next section, here we are approximating the joint probability of the entire path (that is, $s^*_\theta(x,t) \approx p_t(x(t))$) and split it into a product of conditional distributions involving neighboring points as described in Figure \ref{chain_gen_fig}. This is similar to the approach used in \cite{heng2021simulating}, but without fixing time. Our new loss function significantly reduces the dimension of the neural network optimization problem:
\begin{figure}[ht]
\centering
\includegraphics[width=0.8\textwidth]{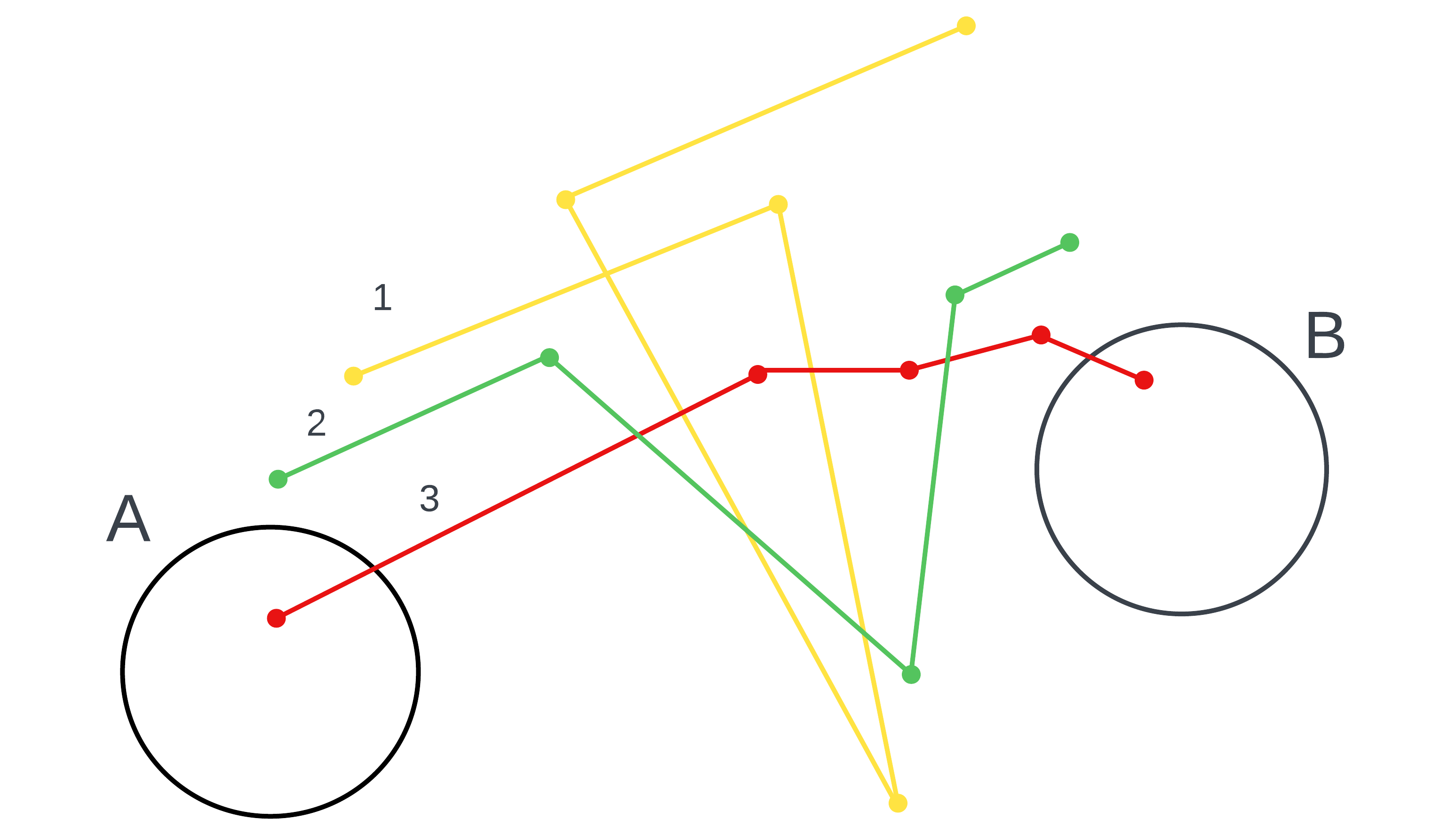}
\caption{\textbf{Chain procedure for generating paths.} Starting with a path that is pure noise (colored yellow), path points are updated together in each step and move towards a higher probability path (colored red). The movement of each point in a single step depends only on the position of its neighbors. }\label{chain_gen_fig}
\end{figure}
\begin{equation}
\begin{split}
L(\theta) = & \frac{1}{N} \sum_{i = 1}^{N} h_{t_i} \mathbb{E}_{x(0)} \mathbb{E}_{x(t_i) \mid x(0)} \Bigl\lVert \sum_{j = 1}^m \frac{\partial}{\partial x_j} \log p_{t_i} (x(t_i) \mid x(0), T^*) - s^*_\theta(x(t_i), t_i, T^*)^{(j)}\Bigr\rVert_2^2,
\end{split}
\end{equation}
where $x_j$ now denotes the $j$th point in the path. This is a slight change in notation from the previous section, where the subscript was the time in the generation process. The time flow of generation is now represented by $x(t)$. We can take advantage of the Markov property of transition paths to get a simplified expression for the sub-score for interior points,
\begin{equation}
\begin{split}
    & p(x) = p(x_1) \prod_{i} p(x_i \mid x_{i-1}) \\
    \implies & \log p(x) = \log p(x_i) + \sum_{i} \log p(x_i \mid x_{i-1}) \\
    \implies & \frac{\partial}{\partial x_n} p(x) = \frac{\partial}{\partial x_n} \log p(x_n \mid x_{n-1}) + \frac{\partial}{\partial x_n} \log p(x_{n+1} \mid x_n).
\end{split}
\end{equation}

It follows that we need two networks, $s_{\theta_1}(x_n(t), x_{n-1}(t), n, t, T^*)$ and $s_{\theta_2}(x_n(t), x_{n+1}(t), n, t, T^*)$. The first to approximate $\frac{\partial}{\partial x_n} \log p_t(x_n \mid x_{n-1}, T^*)$ and the second for $\frac{\partial}{\partial x_n} \log p_t(x_{n+1} \mid x_n, T^*)$. The first and last points of the path require slightly different treatment. From the same decomposition of the joint distribution, we have that
\begin{align}
\frac{\partial}{\partial x_1} p(x) & = \nabla \log p(x_1) + \frac{\partial}{\partial x_1} \log p(x_2 \mid x_1), \\
\frac{\partial}{\partial x_m} p(x) &= \frac{\partial}{\partial x_m} \log p(x_m \mid x_{m-1}).
\end{align}
Then, the entire description of the sub-score functions is given by
\begin{equation}
    s_\theta^*(x, t)^{(n)} = 
    \begin{cases}
        s_{\theta_3}(x_1, t, T^*) + s_{\theta_2}(x_1, x_2, 1, t, T^*) & \text{if } n = 1 \\
        s_{\theta_1}(x_m, x_{m-1}, m, t, T^*) & \text{if } n = m\\
        s_{\theta_1}(x_n, x_{n-1}, n, t, T^*) + s_{\theta_2}(x_n, x_{n+1}, n, t, T^*) & \text{otherwise}
    \end{cases},
\end{equation}
where $s_{\theta_3}(x_1, t, T^*)$ is a third network to learn the distribution of initial points. As outlined in Section \ref{reverse_SDE_sampling}, we use a discretized version of the reverse SDE to generate transition paths.

\subsection{Midpoint Reverse SDE Denoising}
\label{midpoint_generation_sec}

While the previous approach can generate paths well, there are drawbacks. It requires learning $O(2m)$ score functions since the score for a single point is affected by its index along the path. There is a large degree of correlation for parts of the path that are close together, but it is still not an easy computational task. Additionally, all points are updated simultaneously during generation, so the samples will start in a lower data density region. This motivates using a midpoint approach, which is outlined in Figure \ref{midpoint_generation}.

\begin{figure}[H]
\centering
\includegraphics[width=0.8\textwidth]{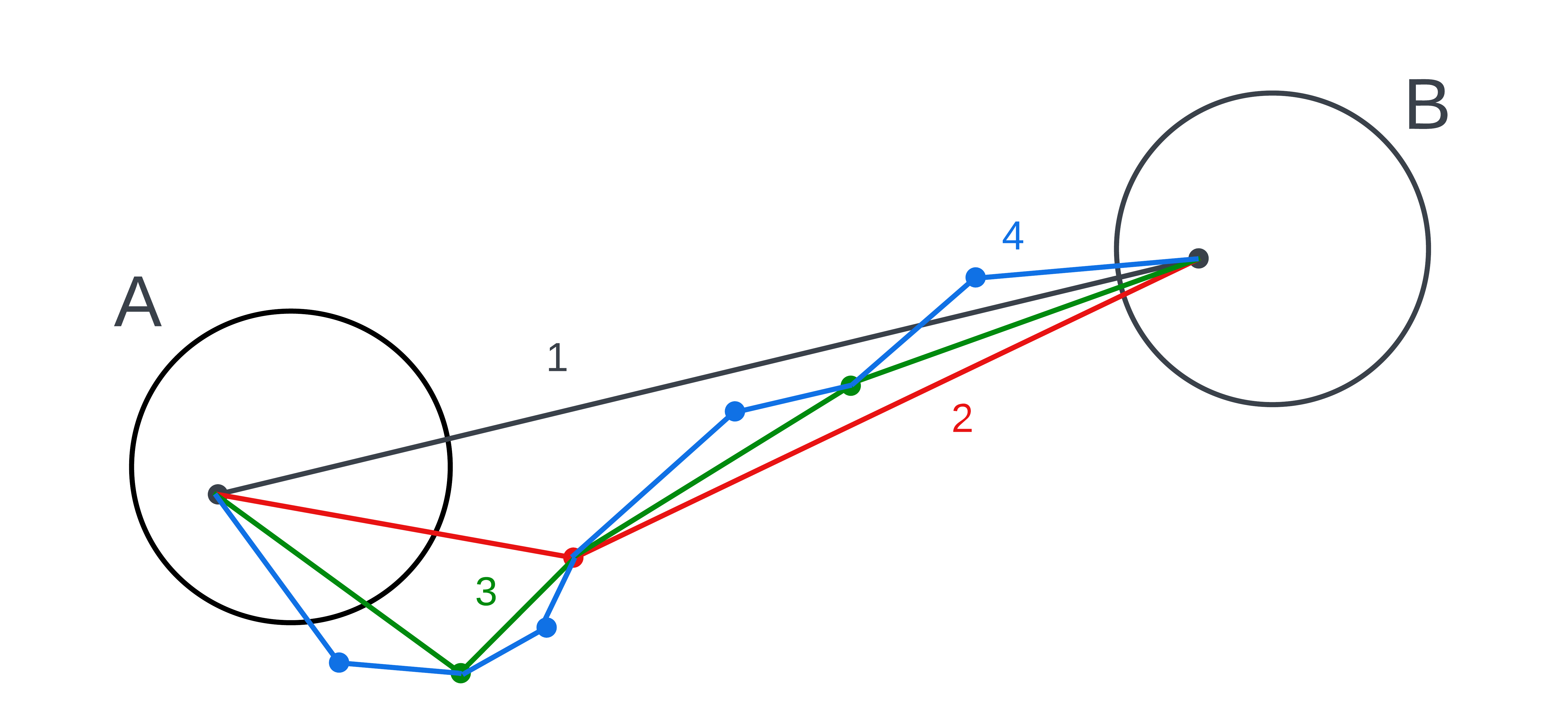}
\caption{\textbf{Midpoint procedure for generating paths.} The two endpoints are generated independently at the start of the procedure. Then, the path's interior is constructed incrementally via splitting. In each iteration, a new point is generated between every adjacent pair of existing points, making the path more refined. }\label{midpoint_generation}
\end{figure}

We train a score model that depends on both endpoints to learn the distribution of the point equidistant in time to each. This corresponds to the following decomposition of the joint density. Let us say that there are $2^k + 1$ points in each path. We start with the case in which the path is represented by a discrete Markov process $x = \{x_j\}_{j=1}^{2^k+1}$.

\begin{lemma}
Define $ f(i) = \max \{ k \in \mathbb{N} : (i - 1) \mod 2^k = 0 \} $, then
\[ p(x) = p(x_1) p(x_{2^k + 1} \mid x_1) \prod_{i=2}^{2^k} p(x_i \mid x_{i - f(i)}, x_{i + f(i)}).\]
\label{midpoint_lemma}
\end{lemma}

\begin{proof}
First, we want to show that $p(x_2, ..., x_{2^k} \mid x_1, x_{2^k+1}) = \prod_{i=2}^{2^k} p(x_i \mid x_{i - f(i)}, x_{i + f(i)})$. We proceed by induction. The inductive hypothesis is that  $p(x_2, ..., x_{2^{k-1}} \mid x_1, x_{2^{k-1}+1}) = \prod_{i=2}^{2^{k-1}} p(x_i \mid x_{i - f(i)}, x_{i + f(i)})$.
\begin{equation*}
\begin{split}
p(x_2, ..., x_{2^k} \mid x_1, x_{2^k+1}) & = p(x_2, ..., x_{2^{k-1} + 1} \mid x_1, x_{2^k+1})p(x_{2^{k-1} + 1}, ..., x_{2^k} \mid x_1, x_2, ..., x_{2^{k-1} + 1}, x_{2^k+1})  \\
& = p(x_2, ..., x_{2^{k-1} + 1} \mid x_1, x_{2^k+1})p(x_{2^{k-1} + 1}, ..., x_{2^k} \mid x_{2^{k-1} + 1}, x_{2^k+1}) \\
& = \prod_{i=2}^{2^k} p(x_i \mid x_{i - f(i)}, x_{i + f(i)}). \qedhere
\end{split}
\end{equation*}
\end{proof}

Based on this, we can parameterize the score function by
\begin{equation}
s_\theta(x_{n^*}, x_{n_1}, x_{n_2}, \frac{n_d}{m} T^*, t) \approx \nabla \log p_t(x_{n^*} \mid x_{n_1}, x_{n_2}),
\end{equation}
where $n_s = \frac{n_1 + n_2}{2}, n_d = n_2 - n_1$. It is worth mentioning that without further knowledge of the transition path process, we would need to include $n_1$ as a parameter for $s_\theta$. However, we can see from \eqref{transition_path_SDE} that transition paths are a time-homogeneous process, which allows this simplification. Then,
\begin{align}
    p(x_{n_s} \mid x_{n_1}, x_{n_2}) & = \frac{p^*(x_{n_s} \mid x_{n_1}) p^*(x_{n_2} \mid x_{n_s})}{p^*(x_{n_2} \mid x_{n_1})} \\
    & = \frac{f(x_{n_s}, x_{n_1}, \frac{n_d}{2}) f(x_{n_s}, x_{n_1}, \frac{n_d}{2})}{f(x_{n_2}, x_{n_1}, n_d)} \\
    & = g(x_{n_1}, x_{n_s}, x_{n_2}, n_d),
\end{align}
where $p^*$ is the transition kernel for transition paths and $f, g$ are general functions designed to show the parameter dependence of $p(x_{n^*} \mid x_{n_1}, x_{n_2})$. The training stage is similar to that described in the previous section. For an interior point $x_i$, the corresponding score matching term is $s_\theta(x_{i}, x_{i-f(i)}, x_{i + f(i)}, \frac{2i_d}{m} T^*, t)$. For the endpoints, we have $s_{\theta_2}(x_{0} \mid T^*)$ and $s_{\theta_3}(x_{2_k} \mid x_0, T^*)$. It is algorithmically easier to use $2^k + 1$ discretization points for the paths, but we can generalize to $k$ points. Starting with the same approach as before, we will eventually end up with a term that can not be split by a midpoint if we want the correct number of points. Specifically, this occurs when we want to generate $n$ interior points for some even $n$. In this case, we can use
\begin{multline}
    p(x_{m_1+1}, x_{m_1+2}, ... x_{m_2-1} \mid x_{m_1}, x_{m_2}) \\= p(x_{m^*}) p(x_{m_1+1}, ..., x_{m^* - 1} \mid x_{m_1}, x_{m^*}) p(x_{m^* + 1}, ..., x_{m_2 - 1} \mid x_{m^*}, x_{m_2}),
\end{multline}
where $m^* = \frac{m_1 + m_2 + 1}{2}$ or $\frac{m_1 + m_2 - 1}{2}$. The score function will require an additional time parameter now that the midpoint is not exactly in between the endpoints. A natural choice is to use 
\begin{equation*}
s_\theta(x_{i}, x_{i-i^*}, x_{i + i^*}, \frac{2i^*}{m} T^*, t_{\text{shift}}, t), 
\end{equation*}
where
\begin{equation}
    t_{\text{shift}} = 
    \begin{cases}
        0 & \text{if midpoint is centered} \\
        \frac{T^*}{2m} & \text{otherwise}
    \end{cases}.
\end{equation}

A similar adaptation can be made when the points along the path are not evenly spaced. It remains to learn $P(T^*)$, which is a simple, one-dimensional problem. We can then use times drawn from the learned distribution as the seed for generating paths. 

We can extend the result from Lemma \eqref{midpoint_lemma} to the continuous case. Consider a stochastic process $x(t)$ as in \eqref{stoc_process} and arbitrary times $t_1, ..., t_n$. 

\begin{corollary}
Define $ f(i) = \max \{ k \in \mathbb{N} : (i - 1) \mod 2^k = 0 \} $
\[ p(x(t_1), ..., x(t_n)) = p(x(t_1)) p(x(t_n) \mid x(t_1)) \prod_{i=2}^{2^k} p(x(t_i) \mid x(t_{i + f(i)}), x(t_{i + f(i)})).\]
\label{midpoint_corollary_cont}
\end{corollary}

The proof follows by the same induction as in the previous case, since the Markov property still holds. This technique can be extended to problems similar to transition path generation, such as other types of conditional trajectories. It is also possible to apply the approach of non-simultaneous generation to the chain method that we described in the previous section.

\subsection{Convergence Guarantees for Score-Matching} \label{conv-guarantees}
We seek to obtain a convergence guarantee for generating paths from the reverse SDE using the approach outlined in Section \ref{reverse_SDE_sampling}. There has been previous work on convergence guarantees for general data distributions by \cite{lee2023convergence}, \cite{chen2023improved}, \cite{de2022convergence}, \cite{chen2022sampling}. In particular, we can make guarantees for the KL divergence or TV distance between $p_{data}$ and the distribution of generated samples given an $L_2$ error bound on the score estimation. Recent results \cite{benton2023linear} show that a bound with linear $d$-dependence is possible for the number of discretization steps required to achieve $\mathrm{KL}(p_{data} || q_{t_N}) = \tilde O(\epsilon_0^2)$.

\begin{assumption}
The error in the score estimate at the selected discretization points $t_1, t_2, ..., t_k$ is bounded:
\begin{equation}
\frac{1}{T} \sum_{k=1}^N h_k \mathbb{E}_{p_{t_k}}||s_\theta(x, t_k) - \nabla \log p_{t_k}||^2 \leq \epsilon_0^2.
\end{equation}
\end{assumption}

\begin{assumption}
The data distribution has a bounded second moment
\begin{equation}
    \mathbb{E}_{p_{data}}||x||^2 < \infty.
\end{equation}
\end{assumption}

With these assumptions on the accuracy of the score network and the second moment of $p_{data}$, we can establish bounds for the KL-divergence between $p_{data}$ and $q_{t_N}$. Using the exponential integrator scheme from Section \ref{reverse_SDE_sampling}, the following theorem from \cite{benton2023linear} holds:

\begin{theorem}
Suppose that Assumptions 1 and 2 hold. If we choose $T = \frac{1}{2} \log \frac{d}{\epsilon_0^2}$ and $N = \Theta(\frac{d(T + \log (\frac{1}{\delta}))^2}{\epsilon_0^2})$, then there exists a choice of $M$ from \ref{reverse_SDE_sampling} such that $\mathrm{KL}(p_{data} || q_{t_N}) = \tilde O(\epsilon_0^2)$.
\end{theorem}

It is important to keep in mind that the KL-divergence between $p_{data}$ and $q_{t_N}$ does not entirely reflect the quality of the generated samples. In a practical setting, there also must be sufficient differences between the initial samples and the output. Otherwise, the model is performing the trivial task of generating data that is the same or very similar to the input data.

\section{Results} \label{results}

\subsection{M\"uller Potential}

To exhibit the effectiveness of our algorithms, we look at the overdamped Langevin equation from \eqref{eq:1.1} and choose $V$ as the two-welled M\"uller potential in $\mathbb{R}^2$ defined by
\begin{equation}
\label{eqn:5.1}
    V(x) = \sum_{i=1}^{4} D_i \text{ exp}(a_i(X_i - x_1)^2 + b_i(X_i - x_1)(Y_i - x_2) + c_i(Y_i - x_2)^2).
\end{equation}
We used the following parameters, as used in \cite{khoo2019solving} \cite{li2019computing}
\begin{alignat*}{2}
    \tag{5.2}
    \label{5.2}
    a = [-1, -1, -6.5, 0.7], \\
    b = [0, 0, 11, 0.6], \\
    c = [-10, -10, -6.5, 0.7], \\
    d = [-200, -100, -170, 15].
\end{alignat*}
We define regions A and B as circles with a radius of 0.1 centered around the minima at (0.62, 0.03) and (-0.56, 1.44) respectively. This creates a landscape with two major wells at A and B, a smaller minimum between them, and a potential that quickly goes to infinity outside the region $\Omega := [-1.5, 1.5] \times [-0.5, 2]$. We used $B^{-1} = 10 \sqrt{2}$ for this experiment, which is considered a moderate temperature.

For each of the score functions learned, we used a network architecture consisting of 6 fully connected layers with 20, 400, 400, 200, 20, and 2 neurons, excluding the input layer. The input size of the first layer was 7 (2 data points, $n, T^*$ and $t$) for the chain method and 8 for the midpoint method (3 data points, $T^*$ and $t$). We used a hyperbolic tangent (tanh) for the first layer and leaky ReLU for the rest of the layers. The weights were initialized using the Xavier uniform initialization method. All models were trained for 200 epochs with an initial learning rate of 0.015 and a batch size of 600. For longer paths, we used a smaller batch size and learning rate because of memory limitations. The training set was about 20,000 paths, so this example is in the data-rich regime. All noise levels were trained for every batch. We used the \eqref{eqn:2.3} to discretize the reverse process. We found that $t_{min} = 0.005$ and $T = 7$ with 100 discretization points gave strong empirical results. For numerical reasons, we used a maximum step size of 1. In particular, this restriction prevents the initial step from being disproportionately large, which can lead to overshooting.

\begin{figure}[H]
\centering
\includegraphics[width=0.8\textwidth]{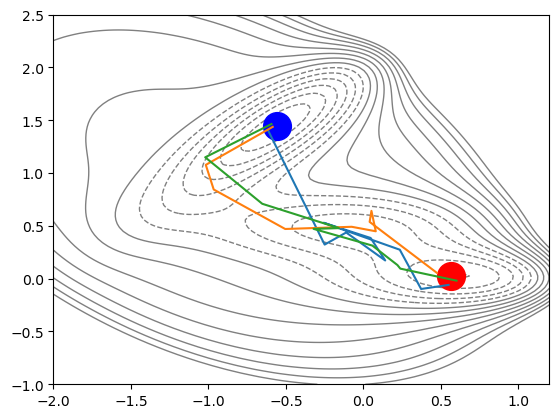}
\caption{\textbf{Sample paths generated using midpoint method.} The paths display consistent qualitative behavior characterized by random fluctuations that remain within low energy levels. The region for minimum A is shown with a red circle and the region for minimum B is shown with a blue circle.}\label{mueller_paths_fig}
\end{figure}

The paths generated using either method (sample shown in Figure \ref{mueller_paths_fig}) closely resemble the training paths. To evaluate the generated samples numerically, we can calculate the relative entropy of the generated paths compared to the training paths or the out-of-sample paths (Tables \ref{mid_entropy}, \ref{chain_entropy}). We converted collections of paths into distributions using the scipy gaussian\_kde function.

\begin{figure}[H]
    \centering
    \subfloat[Relative entropy comparison of NN generated transition paths with training and out of sample paths using midpoint method]{%
    \begin{tabular}[t]{|c|c|c|c|}
        \hline
        \# Discretization Points & Output and Training & Output and OOS & Training and OOS \\
        \hline
        9 & 0.018 & 0.068 & 0.031 \\
        \hline
        17 & .020 & 0.078 & 0.029 \\
        \hline
        33 & 0.032 & 0.109 & 0.035 \\
        \hline
    \end{tabular}
    \label{mid_entropy}
}
\end{figure}

\begin{figure}[H]
    \centering
    \subfloat[Relative entropy comparison of NN generated transition paths with training and out of sample paths using chain method]{%
    \begin{tabular}[t]{|c|c|c|c|}
        \hline
        \# Discretization Points & Output and Training & Output and OOS & Training and OOS \\
        \hline
        9 & 0.012 & 0.056 & 0.031 \\
        \hline
        17 & 0.016 & 0.065 & 0.029 \\
        \hline
        33 & 0.025 & 0.035 & 0.081 \\
        \hline
    \end{tabular}
    \label{chain_entropy}
}
\end{figure}

We can also test the quality of our samples by looking at statistics of the generated paths. For example, we can look at the distribution of points on the path. From transition path theory \cite{lu2015reactive}, we know that $p^*(x) \propto p(x) q(x) (1 - q(x))$, where $p$ is the original probability density of the system, $p^*$ is the density of points along transition paths, and $q$ is the committor. This gives a 2-dimensional distribution as opposed to a $2m$ dimensional one since each of the path points is treated independently. Here, the three relevant distributions are the training, output, and ground truth distributions. 

\begin{figure}
  \centering
  \subfloat[Difference Plot]{\includegraphics[width=0.45\textwidth]{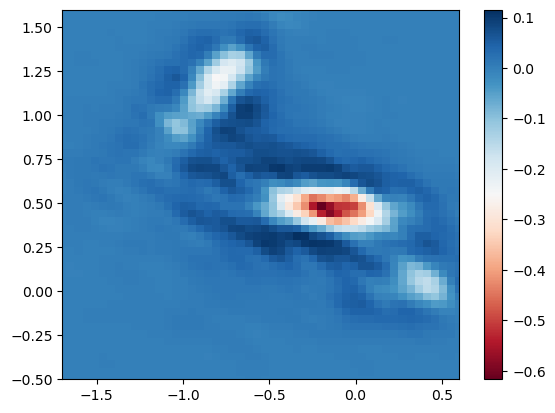}}%
  \qquad
  \subfloat[Training Density]{\includegraphics[width=0.45\textwidth]{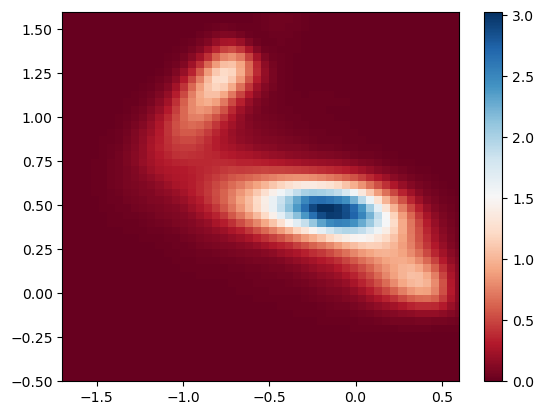}}%
  \caption{\textbf{Comparison of the density path points between training and output samples.} The density of path points for the training sample is shown on the right, while the difference between the densities is shown on the left. The maximum relative difference is around $20\%$.}
\end{figure}

We include the average absolute difference between the three relevant distributions at 2500 grid points in Tables \ref{dist_mid}, \ref{dist_chain}.

\begin{figure}[H]
\centering
    \subfloat[Average absolute difference comparison of the distribution of points on transition paths using midpoint method]{%
        \begin{tabular}[t]{|c|c|c|c|}
         \hline
         \# Discretization Points & Output and Training & Output and True & Training and True \\
         \hline
         9 & 0.027 & 0.055 & 0.034 \\
         \hline
         17 & 0.031 & 0.060 & 0.034 \\
         \hline
         33 & 0.039 & 0.072 & 0.037 \\
         \hline
        \end{tabular}
        \label{dist_mid}
}
\end{figure}

\begin{figure}[H]
\centering
    \subfloat[Average absolute difference comparison of the distribution of points on transition paths using chain method]{%
        \begin{tabular}[t]{|c|c|c|c|}
         \hline
         \# Discretization Points & Output and Training & True and Training & Output and True \\
         \hline
         9 & 0.022 & 0.051 & 0.034 \\
         \hline
         17 & 0.027 & 0.056 & 0.034 \\
         \hline
         33 & 0.034 & 0.062 & 0.037 \\
         \hline
        \end{tabular}
        \label{dist_chain}
}
\end{figure}

The chain method slightly outperforms the midpoint method according to this metric. Qualitatively, the paths generated using the two methods look similar. There is a strong correlation between the error and the training time, but a limited correlation with the size of the training dataset after reaching a certain size. This is promising for potential applications, as the data requirements are not immediately prohibitive.

\begin{figure}[H]
\centering
\includegraphics[width=0.7\textwidth]{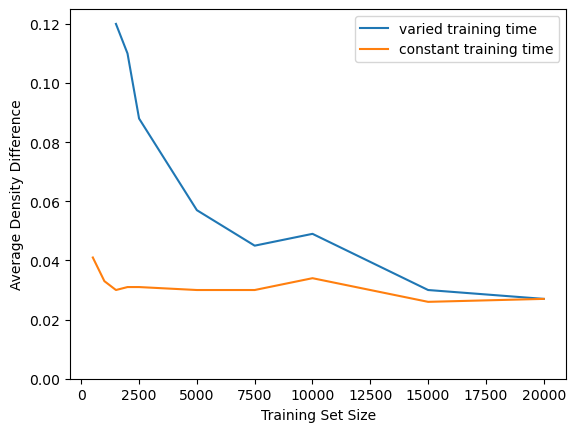}
\caption{\textbf{Comparison of average probability density difference between the 20,000 path dataset and the generated paths for different sizes of the training subset.} When the training scheme is kept constant, but the size of the training set is smaller, the model's accuracy drops drastically. On the other hand, when the number of epochs was adjusted accordingly, the accuracy did not vary significantly until a sample size of around 500.}\label{data_dependency_fig}
\end{figure}

The strong dependence on training time, which is shown in Figure \ref{data_dependency_fig}, supports the notion that the numerical results of this paper could be improved with longer training or a larger model. We further explore performance in the data-scarce case in the following section.

\subsection{Alanine Dipeptide}
For our second numerical experiment, we generate transitions between stable conformers of Alanine dipeptide. The transition pathways \cite{jang2006multiple} and free-energy profile \cite{vymetal2010metadynamics} of Alanine dipeptide have been studied in the molecular dynamics literature and are shown in Figure \ref{reactive_pathways_fig}. This system is commonly used to model dihedral angles in proteins. We use the dihedral angles, $\phi$ and $\psi$ (shifted for visualization purposes), as collective variables.

The two stable states that we consider are the lower density state around $\psi = 0.6 \text{ rad}, \phi = 0.9 \text{ rad}$ (minimum A), and the higher density state around $\psi = 2.7 \text{ rad}, \phi = -1.4 \text{ rad}$ (minimum B). The paths are less regular than with M\"uller potential because the reactive paths "jump" from values close to $\pi$ and $-\pi$. We use 22 reactive trajectories as training data for the score network. The data include 21 samples that follow reactive pathway 1 and one sample that follows reactive pathway 2. We use the same network architecture as used for the midpoint approach in the previous section. We train the model for 2000 epochs with an initial learning rate of 0.0015.

Because the data are angles, the problem moves from $\mathbb{R}$ to the surface of a torus. We wrap around points outside of the interval $[-\pi, \pi]$ during the generation process to adjust for the new topology. This approach was more effective than encoding periodicity into the neural network. For this problem, the drift term for the forward process, $-\frac{1}{2} x$, is no longer continuous, though this did not cause numerical issues. It may be useful to use the more well-behaved process $d\Omega_t = \sin (-\Omega_t) dt + \sqrt{2} dW_t$ applied to problems with this geometry.

\begin{figure}[H]
\centering
\includegraphics[width=0.8\textwidth]{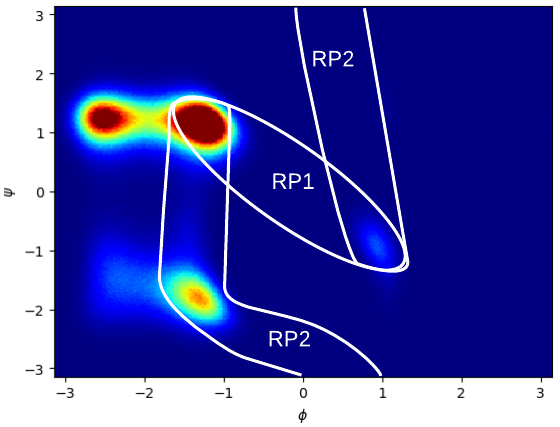}
\caption{\textbf{Alanine Dipeptide probability density and reactive pathways.} The trajectories are highly concentrated around the bottom minimum. Transitions through the first reactive pathway (RP1) occur much more frequently than transitions through the second reactive pathway (RP2).} \label{reactive_pathways_fig}
\end{figure}

\begin{figure}[ht]
  \centering
  \subfloat{\includegraphics[width=0.45\textwidth]{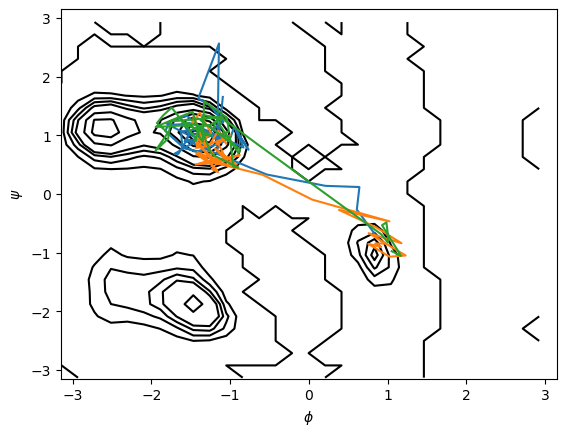}} \hfill
  \subfloat{\includegraphics[width=0.45\textwidth]{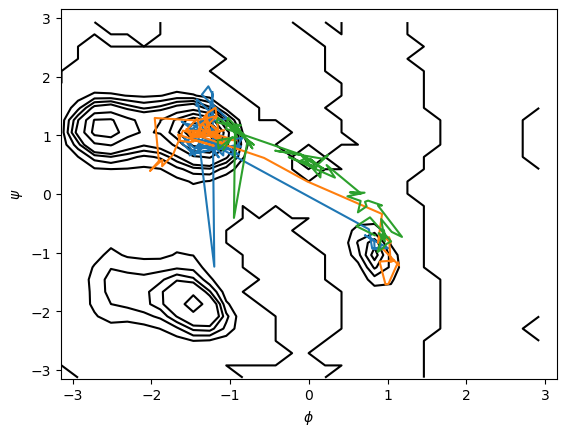}} \\
  \subfloat[Training Paths]{\includegraphics[width=0.45\textwidth]{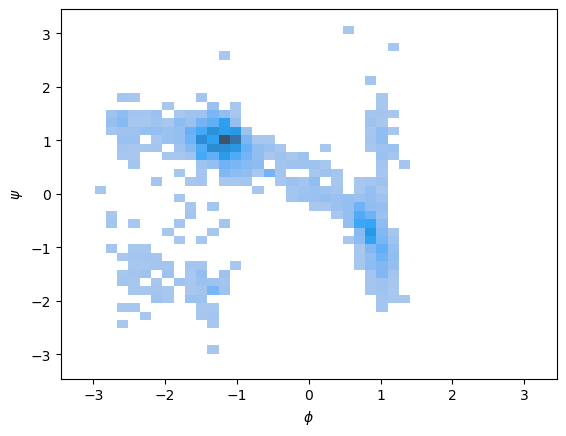}}
  \hfill
  \subfloat[Generated Paths]{\includegraphics[width=0.45\textwidth]{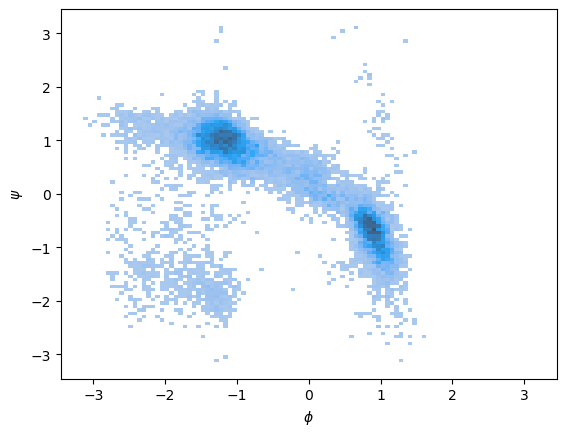}}
  \caption{\textbf{Comparison of original paths with generated paths using dihedral angles as reactive coordinates.} Pictured on the top are a few paths from the training set (left) and generated using the midpoint method (right). On the bottom, the densities of points on the paths are shown for both cases. }
  \label{ala_paths_fig}
\end{figure}

We can generate paths across both reactive channels. In some trials, paths are generated that follow neither channel, instead going straight from the bottom left basin to minimum A. This occurred with similar frequency to paths across RP2. It is more challenging to determine the quality of the generated paths when there is minimal initial data, but it is encouraging that the qualitative shape of the paths is similar (Figure \ref{ala_paths_fig}). Specifically, the trajectories tend to stay around the minima for a long time, and the transition time between them is short. For this experiment, the paths generated using the midpoint method were qualitatively closer to the original distribution. In particular, disproportionately large jumps between adjacent points occurred more frequently using the chain method, though both methods were able to generate representative paths. 

\FloatBarrier 
\bibliography{references}{}
\bibliographystyle{alpha}

\end{document}